\def\draft{0} 
\def\squeezed{1} 
\def\full{1} 

\documentclass[11pt]{article}
\usepackage{amsfonts,amsmath,amssymb,amsthm,boxedminipage,color,url,fullpage}

\usepackage[numbers]{natbib} 
\usepackage[latin1]{inputenc}
\usepackage{times}
\usepackage[hylinks,draft,titlepage,proceed,nofullpage,usetoc]{tcs}
\usepackage{hyperref}

\usepackage{color}
\definecolor{myblue}{rgb}{0,0,0.5}

\ifnum\squeezed=1
\usepackage{sectsty}
\usepackage[
	hmargin=0.95in,vmargin=1.0in]{geometry}
\fi

\newcommand{\zo}{\{0,1\}}
\newcommand{\eps}{\epsilon}
\newcommand{\dans}{\to}
\newcommand{\ie}{\emph{i.e.} }
\newcommand{\Exp}{\mathbb{E}}

\newcommand{\IC}{\mathsf{IC}}
\newcommand{\err}{\mathsf{err}}

\newcommand{\srec}{\mathsf{srec}}
\newcommand{\prt}{\mathsf{prt}}
\newcommand{\rect}{\mathsf{rec}}
\newcommand{\bprt}{\bar{\mathsf{prt}}}

\newcommand{\calU}{\mathcal{U}}
\newcommand{\Bad}{B}
\newcommand{\oB}{\overline{\Bad}}
\newcommand{\Al}{\mathsf{A}}
\newcommand{\Bo}{\mathsf{B}}
\newcommand{\calA}{\mathcal{A}}
\newcommand{\calB}{\mathcal{B}}
\newcommand{\calX}{\mathcal{X}}
\newcommand{\calY}{\mathcal{Y}}

\newcommand{\rel}{\mathcal{R}}
\newcommand{\setin}{\mathcal{I}}
\newcommand{\setout}{\mathcal{Z}}

\newcommand{\VSP}{\mathrm{VSP}}
\newcommand{\tVSP}{\widetilde{\mathrm{VSP}}}
\newcommand{\GHD}{\mathrm{GHD}}
\newcommand{\ORT}{\mathrm{ORT}}

\newcommand{\comment}[1]{}

\newcommand{\alea}{r_\pi}
\newcommand{\alean}{r_{\pi'}}

\usepackage[compact]{titlesec}
\titlespacing{\section}{0pt}{4pt}{4pt}
\titlespacing{\subsection}{0pt}{4pt}{4pt}
\titlespacing{\subsubsection}{0pt}{4pt}{4pt}

\usepackage{mdwlist}

\author{Iordanis Kerenidis   \thanks{CNRS, LIAFA, Universit\'e Paris 7 and CQT, NUS Singapore.  \url{jkeren@liafa.univ-paris-diderot.fr}} \and 
	Sophie Laplante \thanks{LRI, Universit\'e Paris-Sud 11. \url{laplante@lri.fr}} \and 
	Virginie Lerays  \thanks{LRI, Universit\'e Paris-Sud 11. \url{virginie.lerays@lri.fr}} 
  \and J\'er\'emie Roland \thanks{Universit\'e Libre de Bruxelles, QuIC, Ecole Polytechnique de Bruxelles. \url{jroland@ulb.ac.be}}
  \and David Xiao \thanks{CNRS, LIAFA, Universit\'e Paris 7.  \url{dxiao@liafa.univ-paris-diderot.fr}}}

\begin{document}

\title{Lower bounds on information complexity \\via zero-communication protocols \\
and applications
}

\maketitle

\thispagestyle{empty}
\setcounter{page}{0}

\begin{abstract}
  We show that almost all known lower bound methods for communication
  complexity are also lower bounds for the information complexity.  In particular, we
  define a relaxed version of the \emph{partition bound} of Jain and
  Klauck \cite{JK10} and prove that it lower bounds the information
  complexity of any function.  Our relaxed partition bound subsumes
  all norm based methods (e.g. the $\gamma_2$ method) and
  rectangle-based methods (e.g. the rectangle/corruption bound, the
  smooth rectangle bound, and the discrepancy bound), except the
  partition bound.

  Our result uses a new connection between rectangles
  and \emph{zero-communication} protocols where the players can either
  output a value or abort.  We prove the following compression lemma:
  given a protocol for a function $f$ with information complexity $I$,
  one can construct a zero-communication protocol that has non-abort
  probability at least $2^{-O(I)}$ and that computes $f$ correctly
  with high probability conditioned on  not aborting. Then, we show how such a
  zero-communication protocol relates to the relaxed partition
  bound.

  We use our main theorem to resolve three of the open questions
  raised by Braverman \cite{Bra11b}.  First, we show that the
  information complexity of the Vector in Subspace Problem \cite{KR11}
  is $\Omega(n^{1/3})$, which, in turn, implies that there exists an
  exponential separation between quantum communication complexity and
  classical information complexity. Moreover, we provide an
  $\Omega(n)$ lower bound on the information complexity of the Gap
  Hamming Distance Problem.

\end{abstract}
\newpage

\section{Introduction}

Information complexity is a way of measuring the amount of information Alice and Bob must reveal to each other in order to solve a distributed problem. The importance of this notion has been made apparent in recent years through a flurry of results that relate the information complexity of a function and its communication complexity.  One of the main applications of information complexity is to prove direct sum theorems in communication complexity, namely to show that computing $k$ copies of a function costs $k$ times the communication of computing a single copy.
Chakrabarti, Shi, Wirth and Yao \cite{CSWY01} used information complexity to prove a direct sum theorem for simultaneous messages protocols (their notion is now usually called the \emph{external} information complexity, whereas in this paper we work exclusively with what is often called the \emph{internal} information complexity). Bar-Yossef et al.~\cite{BYJKS04}, used the information cost in order to prove a linear lower bound on the two-way randomized communication complexity of Disjointness. More recently, information-theoretic techniques enabled the proof of the first non-trivial direct sum result for general two-way randomized communication complexity: the randomized communication complexity of $k$ copies of a function $f$ is at least $\sqrt{k}$ times the randomized communication complexity of~$f$ \cite{BBCR10}. Then, Braverman and Rao~\cite{BR11}, showed a tight relation between the amortized distributional communication complexity of a function and its internal information cost. Braverman \cite{Bra11b}, defined interactive information complexity, a notion which is independent of the prior distribution of the inputs and proved that it is equal to the amortized communication complexity of the function.  Braverman and Weinstein \cite{BW11} showed that the information complexity is lower bounded by discrepancy. 



The main question pertaining to information complexity is its relation to communication complexity.  On the one hand, the information complexity provides a lower bound on the communication complexity of the function, since there cannot be more information leaked than the length of the messages exchanged. However, it is still open, whether the information complexity of a function can be much smaller than its communication complexity or whether the two notions are basically equivalent. In order to make progress towards this question, it is imperative to provide strong lower bounds for information complexity, and more specifically to see whether the lower bound methods for communication complexity can be compared to the model of information complexity. 

Lower bound methods in communication complexity can
be seen to fall into three main categories: the norm 
based methods, such as the $\gamma_2$ method of Linial and Shraibman~\cite{LS09}
(see Lee and Shraibman's survey for an overview~\cite{LS09b});
the rectangle based methods, such as discrepancy and the rectangle bound;
and, of course, the information theoretic methods, among which, information complexity.
Recently, Jain and Klauck \cite{JK10} introduced the smooth rectangle bound, as well as the stronger partition bound, and showed
that they subsume both $\gamma_2$ and the rectangle bound~\cite{JK10}. 

The first lower bound on information complexity was proved by Braverman \cite{Bra11b}, who showed that it is lower bounded by the logarithm of the communication complexity. Recently, Braverman and Weinstein showed that the discrepancy method lower bounds the information complexity~\cite{BW11}.
Their result follows from a compression lemma for protocols: a protocol for a function $f$ that leaks $I$ bits of information implies the existence of a protocol with communication complexity $O(I)$ and advantage on computing $f$ (over a random guess) of $2^{-O(I)}$.  


\subsection{Our results}

In this paper, we show that all known lower bound methods for communication complexity, with the notable exception of the partition bound, generalize to information complexity. More precisely, we introduce the \emph{relaxed partition bound} (in \autoref{def:relaxed-prt}) denoted by $\bprt^\mu_\eps(f)$, which depends on the function to be computed $f$, the input distribution $\mu$, and the error parameter $\eps$, and
such that the distributional communication complexity
$D^\mu_\eps(f) \geq \log(\bprt^\mu_\eps(f))$ for any $f$.   We prove that the information complexity of a function $f$ is bounded below by the relaxed partition bound:

\begin{theorem}
\label{thm:mainthm}
There is a positive constant $C$ such that for all functions $f:\setin\to\setout$, all $\eps,\delta \in (0, \tfrac{1}{2}]$, and all distributions $\mu$, we have $\IC_\mu(f, \eps) \geq \tfrac{\delta^2}{C} \cdot \left(\log \bprt^{\mu}_{\eps+3\delta}(f)-\log|{\setout}| \right) - \delta$.
\end{theorem}
Since we show in \autoref{claim:srec-bprt-prt} that the relaxed partition bound subsumes the norm based methods (e.g. the $\gamma_2$ method) and the rectangle-based methods (e.g. the rectangle/corruption bound, the smooth rectangle bound, and the discrepancy bound), all of these bounds are also lower bounds on the information complexity. Moreover, together with the direct sum theorem for information complexity, our main result implies a direct sum theorem on communication complexity for many notable functions (see \autoref{cor:dst}).

\paragraph{Technique.} The key idea of our result is a new connection between communication rectangles and zero-communication protocols, where the players can either output a value or abort but \emph{without communicating}. A priori, it is surprising that protocols with no communication can actually provide some insight on the communication or information complexity of a function. However, this model, which has been extensively used in quantum information for the study of non-local games and Bell inequalities, turns out to be  a very powerful tool for the study of classical communication and information complexity. The communication complexity of simulating distributions is known to be
related to the probability of not aborting in zero-communication protocols that can abort~\cite{GG99,M01,BHMR03,BHMR06}.  
More recently connections have been shown for specific lower bound methods.  It has been shown that zero-communication protocols with error give rise to the factorization norm method~\cite{DKLR11}, and the connection between the partition bound and zero-communication protocols with abort was studied in \cite{LLR12}. 
 
In a deterministic zero-communication protocol with abort, each of the two players looks at their input and decides either to abort the protocol or to output some value $z$.  The output of the protocol is $z$ if both players agree on $z$, or it aborts otherwise.
It is easy to see that for any deterministic zero-communication protocol with abort, the set of inputs where both players choose to output $z$ forms a rectangle, and so the protocol is characterized by a set of rectangles each labeled by an output.  In a randomized protocol, we have instead a distribution over labeled rectangles. 

This connection between rectangles and zero-communication protocols with abort allows us to obtain our lower bound for information complexity from a new compression lemma for protocols (\autoref{lem:compress}):  a protocol for a function $f$ that leaks $I$ bits of information implies the existence of a {\em zero}-communication protocol that has non-abort probability at least $2^{-O(I)}$ and that computes $f$ correctly with high probability when not aborting. Our main theorem follows from this new compression. 

The technical tools we use are drawn from Braverman \cite{Bra11b} and in particular Braverman and Weinstein \cite{BW11}.  We describe the difference between our compression and that of \cite{BW11}. There, they take a protocol for computing a function $f$ that has information cost $I$ and compress it to a protocol with communication $O(I)$ and advantage of computing $f$ of $2^{-O(I)}$ (\ie the error increases considerably).  
Then, they apply the discrepancy method, which can handle such small advantage.

In our compression, we suppress the communication entirely, and, moreover, we only introduce an arbitrarily small error since the compressed protocol aborts when it does not believe it can correctly compute the output. 
This compression enables us to provide much sharper lower bounds on the information complexity and in particular, the lower bound in terms of the relaxed partition bound.

\paragraph{Applications.} Our lower bound implies that for most functions for which there exists a lower bound on their communication complexity, the same bound extends to their information complexity. Specifically, we can apply our lower bound in order to resolve three of the open questions in~\cite{Bra11b}.

First, we show that there exists a function $f$, such that the quantum communication complexity of $f$ is exponentially smaller than the information complexity of $f$ (Open Problem 3 in~\cite{Bra11b}).
\begin{theorem}\label{VSP}
There exists a function $f$, s.t. for all $\epsilon \in (0, \tfrac{1}{2}), Q(f,\epsilon) = O(\log(\IC(f,\epsilon)))$.
\end{theorem}

In order to prove the above separation, we show that the proof of the lower bound on the randomized communication complexity of the Vector in Subspace Problem ($\tVSP$) \cite{KR11} provides, in fact, a lower bound on the relaxed partition bound. By our lower bound, this implies that $\IC(\tVSP_{\theta,n},1/3)=n^{1/3}$ (Open Problem 7 in~\cite{Bra11b}). Since the quantum communication complexity of $\tVSP_{\theta,n}$ is $O(\log n)$, we have the above theorem. 
Moreover, this implies an exponential separation between classical and quantum information complexity. We refrain from defining quantum information cost in this paper (see \cite{JN10} for a definition), but since the quantum information cost is always smaller than the quantum communication complexity, the separation follows trivially from the above theorem. 

In addition, we resolve the question of the information complexity of the Gap Hamming Distance Problem ($\GHD$) (Open Problem 6 in \cite{Bra11b}), since the lower bounds on the randomized communication complexity of this problem go through the rectangle/corruption bound~\cite{She11} or smooth rectangle bound \cite{CR10,Vid11}.

\begin{theorem}
\label{thm:GHD}
$\IC(\GHD_n,1/3)= \Omega(n).$
\end{theorem}

Regarding direct sum theorems, it was shown \cite{Bra11b} that the information complexity satisfies a direct sum theorem, namely $\IC_{\mu^k}(f^k, \eps) \geq k \cdot \IC_\mu(f, \eps)$. If in addition it holds that $D^\mu_{\eps'}(f) = O(\IC_\mu(f, \eps))$, then we can immediately deduce that $D^{\mu^k}_\eps(f) \geq \IC_{\mu^k}(f^k, \eps) \geq k \cdot \IC_\mu(f, \eps) \geq \Omega(k \cdot D^\mu_{\eps'}(f))$, \ie the direct sum theorem holds for $f$. Therefore our main result also gives the following corollary:
\begin{corollary}
\label{cor:dst}
For any $\eps, \mu$ and any $f : \setin \rightarrow \setout$, if $D^\mu_{\eps}(f) = O(\log \bprt^\mu_{\eps}(f))$, then for all $\delta > 0$ and integers $k$, it holds that $D^{\mu^k}_\eps(f) \geq \Omega\left(k \cdot \delta^2 (D^\mu_{\eps+3\delta}(f) - \log |\setout|) - k\delta) \right)$.
\end{corollary}

For example, since $D^\mu_\eps(\GHD) \leq n$ holds trivially, this corollary along with the fact that $\log \bprt^\mu_\eps(\GHD) = \Omega(n)$ (\cite{She11,CR10,Vid11}, see \autoref{sec:ghd}) immediately implies a direct sum theorem for $\GHD$.

Finally, regarding the central open question of whether or not it is possible to compress communication down to the information complexity for any function, we note that our result says that if one hopes to prove a negative result and separate information complexity from communication complexity, then one must use a lower bound technique that is stronger than the relaxed partition bound.  To the best of our knowledge, the only such technique in the literature is the (standard) partition bound.  We note, however, that to the best of our knowledge there are no known problems whose communication complexity can be lower-bounded by the partition bound but not by the relaxed partition bound.


\subsection{Related work}

Definitions of information complexity with some variations extend back to the work on privacy in interactive protocols \cite{BCKO93}, and related definitions in the privacy literature appear \cite{Kla04,FJS10,ACC+12}.  Information complexity as a tool in communication complexity was first used to prove direct sum theorems in the simultaneous message model \cite{CSWY01}, and subsequently to prove direct sum theorems and to study amortized communication complexity as stated in the first paragraph of this paper \cite{BYJKS04,BBCR10,BR11,Bra11b,BW11}.  There are many other works using information complexity to prove lower bounds for specific functions or to prove direct sum theorems in restricted models of communication complexity, for example \cite{JKS03, JRS03, JRS05, HJMR07}.

In independent and concurrent work, Chakrabarti et al. proved that information complexity is lower bounded by the smooth rectangle bound under product distributions \cite{CKW12}.  While our result implies the result of \cite{CKW12} as a special case, we note that their proof uses entirely different techniques and may be of independent interest.


\section{Preliminaries}\label{sec:preliminaries}

\subsection{Notation and information theory facts}
Let $\mu$ be a probability distribution over a (finite) universe
$\calU$.  We will often treat $\mu$ as a function $\mu : 2^\calU \dans
[0, 1]$.  For $T, S \subseteq \calU$, we let $\mu(T \mid S) = \Pr_{u
  \leftarrow \mu}[ u \in T \mid S ]$.  For singletons $u \in U$, we write
interchangeably $\mu_u = \mu(u) = \mu(\{u\})$.
Random variables are written in uppercase and fixed values in lowercase.
We sometimes abuse notation and write a random variable in place of the
distribution of that random variable.

For two distributions $\mu, \nu$, we let $|\mu - \nu|$ denote their
statistical distance, \ie $|\mu - \nu| = \max_{T \subseteq \calU}
(\mu(T) - \nu(T))$.  We let $D(\mu \parallel \nu) = \Exp_{U \sim\mu} [ \log
\frac{\mu(U)}{ \nu(U)} ]$ be the relative entropy (\ie KL-divergence).
For two random variables $X,Y$, the mutual information is defined as
$I(X:Y)=H(X)-H(X \mid Y)=H(Y)-H(Y \mid X)$, where $H(\cdot)$ is the Shannon entropy.

A \emph{rectangle} of
$\calX\times\calY$
is a product set $A \times B$ where
$A\subseteq\calX$ and $B\subseteq\calY$.
We let $R$ denote a rectangle in
$\calX\times\calY$.  We let $(x, y) \in \calX \times \calY$ denote a
fixed input, and $(X, Y)$ be random inputs sampled according to some
distribution (specified from context and usually denoted by $\mu$).

\subsection{Information complexity}
We study 2-player communication protocols for calculating a function
$f : \setin\to \setout$, where $\setin\subseteq\calX\times\calY$.  Let
$\pi$ be a randomized protocol (allowing both public and private
coins, unless otherwise specified). We denote the randomness used by
the protocol $\pi$ by $r_\pi$. Let $\pi(x, y)$ denote its output, \ie
the value in $\setout$ the two parties wish to compute.

The transcript of a protocol includes all messages exchanged, the
output of the protocol (in fact we just need that both players can compute the output of the protocol from the transcript), as well as any public coins (but no private
coins).  The complexity of $\pi$ is the maximum (over all random
coins) of the number of bits exchanged.

Let $\mu$ be a distribution over $\calX\times\calY$.
Define
 $\err_f(\pi; x, y) = \Pr_{\alea} [ f(x,y) \neq \pi(x, y)]$ if  $(x,y) \in \setin$ and  $0$ otherwise %
 and $\err_f(\pi; \mu) = \Exp_{(X, Y) \sim \mu} \err_f(\pi; X,
 Y) = \Pr_{\alea, (X, Y) \sim \mu} [ (X,Y)\in\setin \wedge f(X, Y) \neq \pi(X, Y)]$.

\begin{definition}
  Fix $f, \mu, \eps$.  Let $(X, Y, \Pi)$ be the tuple distributed according to
  $(X, Y)$ sampled from $\mu$ and then $\Pi$ being the transcript of
  the protocol $\pi$ applied to $X, Y$.  Then define:
  \begin{enumerate}
  \item $\IC_\mu(\pi) = I(X ; \Pi \mid Y) + I(Y; \Pi \mid X)$
  \item $\IC_\mu(f, \eps) = \inf_{\pi : \err_f(\pi; \mu) \leq \eps} \IC_\mu(\pi)$
  \item $\IC_D(f, \eps) = \max_\mu \IC_\mu(f, \eps)$
   \end{enumerate}
\end{definition}

Braverman \cite{Bra11b} also defined the non-distributional
information cost $\IC$, and all of our results extend to it trivially
by the  inequality $\IC_D \leq \IC$.  (We do not require
the reverse inequality $\IC \leq O(\IC_D)$, whose proof is non-trivial
and was given in \cite{Bra11b}).


\section{Zero-communication protocols and the relaxed partition bound}
\label{sec:compression}


\subsection{The zero-communication model and rectangles}

Let us consider a (possibly partial) function $f$. We say that $(x,y)$ is a valid input if $(x,y)\in \setin$, that is, $(x,y)$ satisfies the promise. 
In the zero-communication model with abort, the players either output a value $z\in\setout$ (they \emph{accept} the run) or output $\bot$ (they \emph{abort}).  
\begin{definition}
  The \emph{zero-communication} model with abort is defined as
  follows:
  \begin{description*}
  \item[Inputs] Alice and Bob receive inputs $x$ and $y$ respectively.
  \item[Output] Alice outputs
$a\in\setout\cup\{\bot\}$
and Bob outputs
$b\in\setout\cup\{\bot\}.$ If both Alice and Bob output the same $z
\in \setout$, then the output is $z$.  Otherwise, the
output is $\bot$.
  \end{description*}
\end{definition}
We will study (public-coin) randomized \emph{zero-communication}
protocols for computing functions in this model.

\subsection{Relaxed partition bound}
\label{sec:relaxed-prt}

The relaxed partition bound with error $\eps$ and input distribution
$\mu$, denoted by $\bprt_\eps^\mu(f)$, is defined as follows.
\begin{definition}
\label{def:relaxed-prt}
The distributional relaxed partition bound $\bprt^\mu_\eps(f)$ is the
value of the following linear program.  (The value of $z$ ranges over
$\setout$ and $R$ over all rectangles, including the empty rectangle.)
\begin{align}
\bprt_\eps^\mu(f)& =  \min_{\eta,p_{R,z}\geq 0}  \; \frac{1}{\eta}
\quad \text{ subject to:} \nonumber \\
 & \sum_{(x,y)\in \setin}\mu_{x,y}\sum_{R:(x,y)\in R} p_{R,f(x,y)} + \sum_{(x,y)\notin \setin}\mu_{x,y}\sum_{z,R:(x,y)\in R} p_{R,z}\geq (1-\eps)\eta \label{eq:partfirst}\\
& \forall (x,y)\in\calX\times\calY,\quad \sum_{z,R:(x,y)\in R} p_{R,z}
\leq \eta \label{eq:partsecond}\\
&\sum_{R,z} p_{R,z}=1.\label{eq:prt-operational}
\end{align}
\comment{
\begin{align}
\bprt_\eps^\mu(f)=\min_{w_{R,z}\geq 0} & \sum_{R,z} w_{R,z} \quad \text{ subject to:}
\nonumber\\
 & \sum_{(x,y)\in \setin}\mu_{x,y}\sum_{R:(x,y)\in R} w_{R,f(x,y)}+\sum_{(x,y)\notin \setin}\mu_{x,y}\sum_{z,R:(x,y)\in R} w_{R,z} \geq 1-\eps \nonumber\\
 &\forall (x,y),\quad \sum_{z,R:(x,y)\in R} w_{R,z} \leq 1.
\end{align}
}
The relaxed partition bound is defined as $\bprt_\eps(f)=\max_\mu \bprt_\eps^\mu(f)$.
\end{definition}

We can identify feasible solutions to the program in
\autoref{def:relaxed-prt} as a particular type of randomized
zero-communication protocol: Alice and Bob sample $(R, z)$ according
to the distribution given by the $p_{R,z}$, and each individually sees
if their inputs are in $R$ and if so they output $z$, otherwise they
abort.  The parameter $\eta$ is the \emph{efficiency} of the protocol
\cite{LLR12}, that is, the probability that the protocol does not
abort, and ideally we want it to be as large as possible.

There is also a natural way to convert any zero-communication protocol
$\pi$ into a distribution over $(R, z)$: sample $z$ uniformly from
$\setout$, sample random coins $\alea$ for $\pi$, and let $R = A \times B$ be
such that $A$ is the set of inputs on which Alice outputs $z$ in the
protocol $\pi$ using random coins $\alea$, and similarly for $B$.
(The sampling of a random $z$ incurs a loss of $|\setout|$ in the
efficiency, which is why our bounds have a loss depending on
$|\setout|$.  See \autoref{sec:mainthmpf} for details.)

\paragraph{Relation to other bounds.} The relaxed partition bound is,
as its name implies, a relaxation of the partition bound
$\prt_\eps(f)$ \cite{JK10}.  We also show that the relaxed partition
bound is stronger than the smooth rectangle bound $\srec^z_\eps(f)$
(the proof is provided in \autoref{app:bprt}).
\begin{lemma}\label{claim:srec-bprt-prt}
For all $f,\eps$ and $z\in\setout$, we have $\srec^z_\eps(f)\leq
\bprt_\eps(f)\leq\prt_\eps(f)$. 
\end{lemma}

Since Jain and Klauck have shown in~\cite{JK10} that the smooth
rectangle bound is stronger than the rectangle/corruption bound, the
$\gamma_2$ method and the discrepancy method, this implies that the
relaxed partition bound subsumes all these bounds as well. Therefore,
our result implies that all these bounds are also lower bounds
for information complexity.

We briefly explain the difference between the relaxed partition bound
and the partition bound
(more details appear in \autoref{app:bprt}).  The partition bound
includes two types of constraints.  The first is a correctness
constraint: on every input, the output of the protocol should be
correct with probability at least $(1-\eps) \eta$.  The second is a
completeness constraint: on every input, the efficiency of the
protocol (\ie the probability it does not abort) should be
\emph{exactly} $\eta$. In the relaxed partition bound, we keep the
same correctness constraint.  Since in certain applications the
function is partial (such as the Vector in Subspace Problem
\cite{KR11}), one also has to handle the inputs where the function is
not defined.  We make this explicit in our correctness constraint.  On
the other hand, we relax the completeness constraint so that the
efficiency may lie anywhere between $(1-\eps)\eta$ and $\eta$.  This
relaxation seems to be crucial for our proof of the lower bound on
information complexity, since we are unable to achieve efficiency exactly $\eta$.

\subsection{Compression lemma}


\begin{lemma}[Main compression lemma]
  \label{lem:compress}
  There exists a universal constant $C$ such that for all distributions $\mu$, communication protocols $\pi$ and       
  $\delta \in (0, 1)$, there exists a zero-communication protocol $\pi'$ 
and a real number $\lambda \geq 2^{-C (\IC_\mu(\pi) /
  \delta^2 + 1/\delta)}$
such that 
%
\begin{equation}
\label{eq:compressdist}
 \left|(X,Y,\pi(X,Y))-(X,Y,\pi'(X,Y)|\pi'(X,Y)\neq\bot)\right|\leq \delta
\end{equation}
(in statistical distance) and 
\begin{eqnarray}
\forall (x,y) \quad \Pr_{\alean}[\pi'(x,y)\neq\bot]\leq (1 + \delta)
\lambda \label{eq:compressbadbot} \\ 
  \Pr_{\alean,(X,Y)\sim\mu}[\pi'(X,Y)\neq\bot] \geq (1-\delta)
  \lambda.  \label{eq:compressavgbot} 
\end{eqnarray}
\end{lemma}

Our compression $\pi'$ extends the strategy outlined by
\cite{BW11}.  At a high level, the protocol $\pi'$ does the
following:
\begin{description*}
\item[Sample transcripts] Alice and Bob use their shared randomness to
  repeat $T$ independent executions of an experiment to sample transcripts
  (\autoref{alg:experiment}).  Alice and Bob each decide whether the
  experiment is accepted (they may differ in their opinions).
\item[Find common transcript] Let $\calA$ be the set of accepted experiments
  for Alice, and $\calB$ the set of accepted experiments for Bob.  They
  try to guess an element of $\calA \cap \calB$.  If they find one, they
  output according to the transcript from this experiment.
\end{description*}
We prove our compression lemma in Section \ref{proof}.

\subsection{ Information cost is lower bounded by the relaxed partition bound}
\label{sec:mainthmpf}

We show how our compression lemma implies the main theorem.

\begin{proof}[Proof of \autoref{thm:mainthm}]
 Let $\pi$ be a randomized communication protocol achieving $\IC_\mu(f, \eps)$ 
 and let $\rel$ be the following relation that naturally arises from the function $f$
\begin{align*}
 \rel=\{(x,y,f(x,y)):(x,y)\in \setin\}\cup\{(x,y,z):(x,y)\notin \setin, z\in\setout\}.
\end{align*}

Let us now consider the zero-communication protocol $\pi'$ from
Lemma~\ref{lem:compress}.  As mentioned in \autoref{sec:relaxed-prt},
there is a natural way to identify $\pi'$ with a distribution over
labeled rectangles $(R, z)$: sample $z$ uniformly from $\setout$,
sample $\alea$ and let $R = A \times B$ where $A$ is the set of inputs
on which Alice outputs $z$, and similarly for $B$.  The sampling of
$z$ incurs a loss of $|\setout|$ in the efficiency.

We make this formal: for any fixed randomness $r$ occurring with
probability $p_r$, we define the rectangle $R(z,r)$ as the set of
$(x,y)$ such that the protocol outputs $z$, and we let
$p_{R,z}=\sum_{r:R=R(z,r)} p_r/|\setout|$.

We check the normalization constraint
\begin{align*}
 \sum_{R,z} p_{R,z}=\frac{1}{|\setout|}\sum_{R,z}\sum_{r:R=R(z,r)} p_r=\frac{1}{|\setout|}\sum_{r}p_r\sum_{R,z:R=R(z,r)} 1=\sum_r p_r =1.
\end{align*}

To see that \autoref{eq:partsecond} is satisfied, we have by
definition of $p_{R,z}$ that for any $(x,y)$:
\begin{align*}
\sum_{z,R:(x,y)\in R} p_{R,z} =\frac{1}{|\setout|}\Pr_{\alean}[\pi'(x,y)\neq\bot]\leq
\frac{(1 + \delta) \lambda}{|\setout|}. 
\end{align*}
Finally, to see that \autoref{eq:partfirst} is satisfied, we have
\begin{eqnarray*}
\lefteqn{\sum_{(x,y)\in \setin}\mu_{x,y}\sum_{R:(x,y)\in R} p_{R,f(x,y)}
 + \sum_{(x,y)\notin \setin}\mu_{x,y}\sum_{z,R:(x,y)\in R} p_{R,z}}\\
&=&\frac{1}{|\setout|}\Pr_{\alean,(X,Y) \sim \mu}[(X,Y,\pi'(X,Y))\in \rel]\\
&=& \frac{1}{|\setout|}\Pr_{\alean,(X,Y)\sim\mu}[\pi'(X,Y)\neq\bot]\Pr_{\alean,(X,Y) \sim \mu}[(X,Y,\pi'(X,Y))\in \rel \mid \pi'(X,Y)\neq\bot]\\
&\geq& \frac{1}{|\setout|}(1-\delta)\,\lambda\,\left(\Pr_{\alean,(X,Y) \sim \mu}[(X,Y,\pi(X,Y))\in \rel]-\delta\right)\\
&\geq& \frac{1}{|\setout|}(1-\delta)\,\lambda\,\left(1-\eps-\delta\right)
\quad \geq \quad \tfrac{1}{|\setout|}\,\lambda\,\left(1-\eps-2\delta\right) 
\quad \geq \quad \tfrac{\lambda ( 1+ \delta)}{|\setout|} (1 - \eps - 3\delta)
%
%
\end{eqnarray*}
where for the last line we used the fact that $\pi$ has error $\eps$,
and so 
$\Pr_{\alea,(X,Y) \sim \mu}[(X,Y,\pi(X,Y))\in \rel] \geq 1-\eps$.
This satisfies the constraints in the linear
program~(\autoref{def:relaxed-prt}) for $\bprt^{\mu}_{\eps+3\delta}(f)$
with objective value $\eta= (1 + \delta) \lambda/|\setout|\geq 2^{-C
  (\IC_\mu(\pi)/\delta^2 + 1/\delta)}/|\setout|$. 
\end{proof}

By the definitions of the  information complexity and the relaxed partition bound, we have immediately
\begin{corollary}
\label{cor:maincor}
There exists a universal constant $C$ such that for all functions
$f:\setin\to\setout$, all $\eps,\delta \in (0, 1/2)$, we have
$\IC_D(f, \eps) \geq \frac{\delta^2}{C} [\log
\bprt_{\eps+3\delta}(f)-\log|{\setout}|] - \delta$.
\end{corollary}

\section{The zero-communication protocol}\label{proof}

The zero-communication protocol consists of two stages. First, Alice
and Bob use their shared randomness to come up with candidate
transcripts, based on the a priori information they have on the
distribution of the transcripts given by the information cost of the
protocol.  To do this, they run some sampling experiments and decide
which ones to accept. Second, they use their shared randomness in
order to choose an experiment that they have both accepted.  If anything fails in the course of the protocol, they abort by outputting~$\bot$.

\subsection{Single sampling experiment}
The single sampling experiment is described in
\autoref{alg:experiment} and appeared first in \cite{BW11} (variants
also appeared in \cite{Bra11b} and \cite{BR11}).
Roughly, \autoref{alg:experiment} takes a distribution $\tau$ and two
distributions $\nu_{\Al}, \nu_{\Bo}$ over a universe $\calU$ such that $\nu_{\Al},
\nu_{\Bo}$ are not too far from $\tau$ and tries to sample an element
of $\calU$ that is close to being distributed according to $\tau$.  

Let us informally describe the goal of this  sampling experiment
in our context.
Alice knowing $x$ and Bob $y$ want to sample transcripts according to $\Pi_{x,y}$ which is the distribution over the transcripts of the protocol $\pi$ applied to $(x,y)$.
When inputs $x,y$ are fixed, the probability of a transcript $u$ occurring
is the product of the probabilities of each bit in the transcript.
The product of the probabilities for Alice's bits is 
some function $p_\Al(u)$ which depends on $x$ and the product of the probabilities for Bob's bits is 
some function $p_\Bo(u)$ which depends on $y$ and $\Pi_{x,y}(u)=p_\Al(u) p_\Bo(u)$.
Alice can also estimate  $p_\Bo(u)$ by taking the average over $y$
of  $\Pi_y(u)$.  Call this estimate $q_\Al(u)$; similarly for 
Bob's estimate $q_\Bo(u)$.  Set $\nu_\Al=p_\Al q_\Al$ and
$\nu_\Bo=q_\Bo p_\Bo$.  

The challenge is that Alice and Bob know only $(p_{\Al}, q_{\Al})$ and
$(p_{\Bo},q_{\Bo})$ respectively and do not know $\tau$ (in our setting,
$\tau=\Pi_{x,y}$).  
They use a variant of rejection sampling, in which
Alice will overestimate  $q_{\Al}$ by a factor $2^\Delta$; likewise for Bob.  
Let us define the set of $\Delta$-bad elements with respect to $\tau,
\nu$ as follows:
$$\Bad_\Delta(\tau, \nu) = \{ u \in \calU \mid 2^\Delta \nu(u) <
\tau(u)\}.$$ %
Intuitively, $u$ is bad if $\tau$ gives much more weight to it than
$\nu$.  Observe that if $\tau = p_{\Al} p_{\Bo}, \nu_{\Al} = p_{\Al} q_{\Al}$, then $u
\notin \Bad_\Delta(\tau, \nu_{\Al})$ implies that $2^\Delta q_{\Al}(u) \geq
p_{\Bo}(u)$.

\begin{protocolf}{\label{alg:experiment} Single sampling experiment}
{\small
  Fix a finite universe $\calU$.  Let $p_{\Al}, q_{\Al}, p_{\Bo}, q_{\Bo} : \calU
  \dans [0,1]$ such that $\tau = p_{\Al} p_{\Bo}$, $\nu_{\Al} = p_{\Al} q_{\Al}$, $\nu_{\Bo} =
  p_{\Bo} q_{\Bo}$ are all probability distributions.

  Alice's input: $p_{\Al}, q_{\Al}$.  Bob's input: $p_{\Bo}, q_{\Bo}$.  Common input:
  parameter $\Delta > 0$.
  \begin{enumerate*}
  \item Using public coins, sample $u \leftarrow \calU$, $\alpha, \beta
    \leftarrow [0, 2^\Delta]$.
  \item Alice accepts the run if $\alpha \leq p_{\Al}(u)$ and $\beta \leq
    2^\Delta q_{\Al}(u)$.
  \item Bob accepts the run if $\alpha \leq 2^\Delta q_{\Bo}(u)$ and $\beta
    \leq p_{\Bo}(u)$.
  \item If both Alice and Bob accept, then we say that the experiment is accepted and the output is $u$.
    Otherwise, the output is $\bot$.
  \end{enumerate*}
}
\end{protocolf}

To prove our compression lemma, we use the following claim about the
single sampling experiment.
\begin{claim}
  \label{lem:single}
  Let $B = \Bad_\Delta(\tau, \nu_{\Al}) \cup \Bad_\Delta(\tau,
  \nu_{\Bo})$.  Let $\gamma = \tau(B)$. Then the following holds about
  \autoref{alg:experiment}:
  \begin{enumerate*}
  \item \label{step:alicebob} The probability that Alice accepts equals $\frac{1}{|\calU|2^\Delta}$ and the same for Bob.
 \item \label{step:nofail} The probability that the experiment is accepted is at most
    $\frac{1}{|\calU|2^{2\Delta}}$ and at least $\frac{1 -
      \gamma}{|\calU|2^{2\Delta}}$.
  \item \label{step:stat} Let $\tau'$ denote the distribution of the output of the
    experiment, conditioned on it being accepted.  Then $|\tau -
    \tau'| \leq \gamma$.
  \end{enumerate*}
\end{claim}

Intuitively, this claim says that Alice accepts each single experiment
with probability $\frac{1}{|\calU|2^{\Delta}}$, and also implies that
conditioned on Alice accepting the $i$'th experiment, it is relatively
likely that Bob accepts it.  Therefore, by repeating this experiment
enough times, there is reasonable probability of Alice and Bob both
accepting the same execution of the experiment.  Conditioned on the
experiment accepting, the output of the experiment is distributed
close to the original distribution $\tau$.  In the next section, we
show how to use a hash function to select a common accepting execution
of the experiment out of many executions.


We will use the following lemma that appears in \cite{Bra11b}.
\begin{lemma}[\cite{Bra11b}]
\label{lem:typical}
For all $\tau, \nu, \Delta, \eps$, it holds that
$\tau(\Bad_\Delta(\tau, \nu)) \leq \frac{D(\tau \parallel \nu) + 1}{
  \Delta}$. 
\end{lemma}

\begin{proof}[Proof of \autoref{lem:single}]
We use the arguments first given in \cite{BW11} and prove the items in order.
\begin{enumerate}
\item \textbf{Probability Alice/Bob accepts.}  We do the
  analysis for Alice; the case for Bob is entirely symmetric.  We may
  write:
\begin{align*}
\Pr[\text{Alice accepts}] & = \sum_{u \in \calU}
\frac{1}{|\calU|} \frac{p_{\Al}(u)}{2^\Delta} q_{\Al}(u) 
 = \frac{1}{|\calU|2^\Delta} \sum_{u \in \calU} \nu_\Al(u) 
 = \frac{1}{|\calU|2^\Delta}.
\end{align*}

\item \textbf{Probability of accepting.} First consider
  $u \notin B$.  For such $u$, if $\alpha \leq p_{\Al}(u)$
  then $\alpha \leq 2^\Delta q_{\Bo}(u)$ and also if $\beta \leq
  p_{\Bo}(u)$ then $\beta \leq 2^\Delta q_{\Al}(u)$.  Therefore we
  may write
  \begin{equation}
    \label{eq:notbad}
  \Pr[\text{Experiment outputs $u$}] = \frac{1}{|\calU|}
  \frac{p_{\Al}(u) p_{\Bo}(u)}{2^{2\Delta}} = \frac{\tau(u)}{|\calU|
    2^{2\Delta}}.
  \end{equation}

  Furthermore, for any $u \in \calU$, we may write
  \begin{equation*}
    \Pr[\text{Experiment outputs }u]  = \frac{1}{|\calU|2^{2\Delta}} \cdot
    \min\left\{ p_{\Al}(u), q_{\Bo}(u) 2^\Delta \right\} \cdot \min\left\{ p_{\Bo}(u), 
      q_{\Al}(u) 2^\Delta \right\} \leq
    \frac{\tau(u)}{|\calU|2^{2\Delta}}.
  \end{equation*}
  
  For the upper bound we have:
  $$\Pr[\exists u \in \calU,\ \text{Experiment outputs $u$}] = \sum_{u
    \in \calU} \Pr[\text{Experiment outputs $u$}] \leq \sum_{u \in
    \calU} \frac{\tau(u)}{|\calU|2^{2\Delta}} =
  \frac{1}{|\calU|2^{2\Delta}}.$$ %
  For the lower bound we have
  \begin{align*}
    \sum_{u \in \calU} \Pr[\text{Experiment outputs $u$}] & \geq \sum_{u
      \notin B} \Pr[\text{Experiment outputs $u$}] \\
    & = \sum_{u \notin B} \frac{\tau(u)}{ |\calU|2^{2\Delta}} \quad =
    \quad \frac{1 - \tau(B)}{|\calU|2^{2\Delta}} \quad = \quad \frac{1
      - \gamma}{|\calU|2^{2\Delta} }.
  \end{align*}

\item \textbf{Statistical closeness of $\tau$ and $\tau'$.} Let $\eta$ denote the
  probability that the experiment is accepted.  From the
  previous point, we have that $\eta \in [ \frac{1 -
    \gamma}{|\calU|2^{2\Delta}}, \frac{1}{|\calU| 2^{2\Delta}}]$.  By
  the definition of statistical distance, it suffices to prove that:
  \begin{equation*}
    \forall S \subseteq \calU, \quad \tau(S) - \tau'(S) \leq \gamma.
  \end{equation*}

  We proceed by splitting the elements of $S$ based on whether they
  intersect $B$.
  \begin{align*}
    \tau(S) - \tau'(S) & = \tau(S \cap B) - \tau'(S \cap B) + \tau(S \cap
    \oB)
    - \tau'(S \cap \oB) \nonumber\\
    & \leq \gamma + \tau(S \cap \oB) - \tau'(S \cap
    \oB). 
  \end{align*}
  From \autoref{eq:notbad} we can deduce that $\tau'(S \cap \oB) =
  \frac{\tau(S \cap \oB)}{|\calU| 2^{2\Delta} \eta}$.  Therefore:
  \begin{align*}
    \tau(S \cap \oB) - \tau'(S \cap \oB) & = \tau(S \cap \oB) (1 -
    \tfrac{1}{|\calU| 2^{2\Delta} \eta}).
  \end{align*}
  Since $\eta \leq \frac{1}{|\calU|2^{2\Delta}}$, we have that
  $(1 - \frac{1}{|\calU|2^{2\Delta} \eta}) \leq 0$, which concludes the proof.
\end{enumerate}
\end{proof}

\subsection{Description and analysis of the zero-communication protocol}
  
Let $\mu$ be any distribution on inputs and $\pi$ be any protocol with information complexity $I = \IC_\mu(\pi)$.
Let $(X, Y, \Pi)$ be the joint random variables where $X, Y$
      are distributed according to $\mu$ and $\Pi$ is the distribution
      of the transcript of the protocol $\pi$ applied to $X, Y$ (by slight abuse of notation we use the letter $\Pi$ for both the transcript and its distribution).  Let
      $\Pi_{x,y}$ be $\Pi$ conditioned on $X = x, Y = y$, $\Pi_x$ be
      $\Pi$ conditioned $X = x$, and $\Pi_y$ likewise.

Let $\calU$ be the space of all possible transcripts.  We assume that each transcript contains the output of the protocol. As shown
      in \cite{Bra11b} and described above, Alice can construct functions $p_{\Al}, q_{\Al}:\calU \dans [0,1]$
      and Bob can construct functions $p_{\Bo}, q_{\Bo}:\calU \dans [0,1]$, such that for all $u \in \calU$,
      $\Pi_{x,y}(u) = p_{\Al}(u) p_{\Bo}(u)$, $\Pi_x(u) = p_{\Al}(u) q_{\Al}(u)$, and
      $\Pi_y(u) = p_{\Bo}(u) q_{\Bo}(u)$.

The zero-communication protocol $\pi'$ is described in \autoref{alg:compressed}.  This protocol is an extension of the one in \cite{BW11}, where here Alice uses public coins to
guess the hash function value instead of calculating and transmitting
it to Bob and both players are allowed to abort when they do not believe they can output the correct value.


  \begin{protocolf}{\label{alg:compressed} Zero-communication
       protocol $\pi'$ derived from $\pi$}
{\small
    Alice's input: $x$.  Bob's input: $y$.  Common inputs: $\delta>0,
    I > 0$.

    Set parameters: $\Delta = \frac{4}{\delta} \cdot (\frac{8 \cdot
      I}{\delta} + 1)$ and $T = |\calU| 2^\Delta \ln (8/\delta)$ and
    $k = \Delta + \log (\frac{64}{\delta} \ln (8/\delta)^2)$.
    \begin{enumerate*}
    \item Alice constructs functions $p_{\Al}, q_{\Al}:\calU \dans [0,1]$
      and Bob constructs functions $p_{\Bo}, q_{\Bo}:\calU \dans [0,1]$,
      such that for all transcripts $u \in \calU$,
      $\Pi_{x,y}(u) = p_{\Al}(u) p_{\Bo}(u)$, $\Pi_x(u) = p_{\Al}(u) q_{\Al}(u)$, and
      $\Pi_y(u) = p_{\Bo}(u) q_{\Bo}(u)$.
    \item (\textbf{Run experiments.}) Using public coins, Alice and Bob run 
      \autoref{alg:experiment} $T$ independent times with inputs $p_{\Al}, q_{\Al}, p_{\Bo}, q_{\Bo}$ and  
      $\Delta$.
    \item Let $\calA = \{ i \in [T] : \text{Alice accepts
        experiment }i\}$ and similarly $\calB$ for Bob.  If either set
      is empty, that party outputs the abort symbol $\bot$.
    \item \label{step:findint} (\textbf{Find intersection.})  Using
      public coins, Alice and Bob choose a random function $h:[T]
      \dans \zo^k$ and a random string $r \in \zo^k$.
      \begin{enumerate*}
      \item Alice finds the smallest $i \in \calA$.  If $h(i) \neq r$
        then Alice outputs $\bot$.  Otherwise, Alice outputs in accordance with the transcript of experiment $i$.
        \item Bob finds the smallest $j \in \calB$ such that $h(j) = r$.
        If no such $j$ exists, he outputs $\bot$.  Otherwise, Bob outputs in accordance with the transcript of experiment $j$.
              \end{enumerate*}
    \end{enumerate*}
}
  \end{protocolf}

In order to analyze our protocol, we first define some events and give bounds on their probabilities.

 \begin{definition}
   We define the following events over the probability space of
   sampling $(X, Y)$ according to~$\mu$ and running $\pi'$ on $(X, Y)$
   to produce a transcript $\Pi$:
    \begin{enumerate*}
    \item \textbf{Large divergence.}  $B_D$ occurs if $(X, Y) = (x,
      y)$ such that $D(\Pi_{x,y} \parallel \Pi_x) >
      \frac{8\IC_\mu(\pi)}{\delta}$ or $D(\Pi_{x,y} \parallel \Pi_y) >
      \frac{8\IC_\mu(\pi)}{\delta}$.  We will also let $B_D$ denote the
      set of such $(x, y)$.
      \item \textbf{Collision.} $B_C$ occurs if there exist distinct $i,
      j \in \calA \cup \calB$ such that $h(i) = h(j) = r$.
    \item \textbf{Protocol outputs something.} $H$ occurs if $\pi'(X,
      Y) \neq \bot$.
   \end{enumerate*}
  \end{definition}

The proof of the main compression lemma (\autoref{lem:compress}) uses
the following claim.
  \begin{claim}
    \label{claim:compresspf}
    The probability of the above events are bounded as follows:
    \begin{enumerate*}
    \item \label{item:largediv} The inputs rarely have large
      divergence: $\Pr_{(X,Y)\sim \mu}[B_D] \leq \delta/4$.
\comment{ 
  \item \label{item:empty} $\calA$ is rarely empty: $\forall (x, y),
      \Pr_{\alean}[B_E] \leq \delta / 8$.
        
}      
    \item \label{item:collision} For all $(x, y)$, the hash function
      rarely has a collision: $\Pr_{\alean}[B_C] \leq \frac{\delta}{16} \cdot 2^{-(k+\Delta)}$. 
    
\item \label{item:gooda} For all $(x, y) \notin B_D$, the probability
  of outputting something is not too small: $\Pr_{\alean}[H]  \geq
  (1-\frac{11\delta}{16})2^{-(k+\Delta)}$. 

\item \label{item:nottoogood} For all $(x, y)$ the probability of
  outputting something is not too large: $\Pr_{\alean}[H] \leq (1 +
  \tfrac{\delta}{16})2^{-(k+\Delta)}$.

\item \label{item:statdist} For all protocols $\pi$, input
  distributions $\mu$ and $\delta > 0$, the protocol $\pi'$ in
  \autoref{alg:compressed} satisfies: For all $(x, y) \notin B_D$, let
  $\Pi'_{x,y,H}$ be the distribution of $\pi'(x, y)$ conditioned on
  $H$ (namely, on $\pi'(x, y) \neq \bot$).  Then $| \Pi_{x,y} -
  \Pi'_{x,y,H} | \leq 3\delta / 4$.
\end{enumerate*}
  \end{claim}







\begin{proof}
 In the following, we will frequently use the fact that for all $p,
 \alpha \in [0, 1]$, it holds that $p(1 - \alpha) \geq p - \alpha$.

  We extend the arguments given in \cite{BW11} to prove the items of the claim in order.
  \begin{enumerate}

 \item By the definition of information complexity and the fact that mutual
    information is equal to the expectation of the divergence, we have that for $(X, Y)$ distributed according to
    $\mu$, 
    \begin{align*}
      \IC_\mu(\pi) & = I(X ; \Pi \mid Y) + I(Y ; \Pi \mid X) 
       = \Exp_{(x, y) \leftarrow (X,Y)} [D(\Pi_{x,y} \parallel \Pi_y) +
      D(\Pi_{x,y} \parallel \Pi_x)].
    \end{align*}
    This implies that $\Exp_{(x, y) \leftarrow (X,Y)} [D(\Pi_{x,y} \parallel
    \Pi_x)] \leq \IC_\mu(\pi)$, and since divergence is non-negative
    we have by Markov's inequality that
    $$\Pr_{(x, y) \leftarrow (X,Y)} [D(\Pi_{x,y} \parallel \Pi_x) > 8
    \IC_\mu(\pi) / \delta] \leq \delta / 8.$$ %
    The same argument holds for $D(\Pi_{x,y} \parallel \Pi_y)$ and  by a
    union bound, we have $\Pr_{(X,Y)\sim \mu}[B_D] \leq \delta/4$.

  \item We may write:
    \begin{align*}
 \lefteqn{\Pr_{\alean}[B_C]}\\
 & = \Pr[\exists i \neq j \in
      [T]\text{ s.t. }i \in (\calA \cup \calB), j \in (\calA \cup \calB), h(i) = h(j) = r] \\
      & \leq \sum_{i \neq j \in [T]} \Pr[i \in (\calA \cup \calB) \wedge j \in (\calA \cup
      \calB) \wedge h(i) = h(j) = r ] \\
        & = \sum_{i \neq j} \Pr[i \in (\calA \cup \calB) ]  \Pr[ j \in (\calA \cup \calB)] 
         \Pr[h(i) = h(j) = r ]
      \\ 
      & \leq T^2 \frac{4}{(|\calU|2^\Delta)^2} \cdot \frac{1}{2^{2k}} \\
      & \leq \tfrac{\delta}{16} \cdot 2^{-(k + \Delta)}.
    \end{align*}
    where we have used the independence between
    the trials and independence of the $h$ from the trials, as well as \autoref{step:alicebob}
    of \autoref{lem:single}.

  \item Let us define $G$ to be the event that the smallest $i \in
    \calA$ satisfies $h(i) = r$, and also $i \in \calB$.  (Notice this
    implies that $\calA, \calB$ are both non-empty.)  We have
\[
   \Pr_{\alean}[G] 
    = \Pr[\calA\neq \varnothing] \cdot \Pr[G
   \mid  \calA\neq \varnothing ].
\]
 Observe that an element $i$ is in $\calA$ if and only if
    experiment $i$ is accepted by Alice.  By \autoref{step:alicebob}
    of \autoref{lem:single}, the probability of Alice aborting each
    experiment $i$ is $1 - \frac{1}{|\calU|2^{\Delta}}$.  Since the
    experiments are independent, the probability of Alice aborting
    all experiments is
    $$  \Pr[\calA = \varnothing ] = \left(1 - \tfrac{1}{|\calU|2^{\Delta}} \right)^T \leq e^{-
      \frac{T}{ |\calU| 2^\Delta}} \leq \delta / 8.$$

We assume now that $\calA$ is non empty and we denote by $i$ its first element.
  \begin{align*}\Pr&[G \mid\calA\neq \varnothing ] \\
 &  = \Pr[h(i) = r \mid \calA\neq \varnothing ] \Pr[i \in \calB \mid  i \in \calA  \wedge h(i)=r \wedge \calA\neq \varnothing].
  \end{align*}
  
    For all $j$, the probability that $h(j) = r$ is exactly $2^{-k}$,
    in particular this holds for the first element of $\calA$.
   
    For any $(x, y) \notin B_D$, we have that $D(\Pi_{x,y} \parallel
    \Pi_x) \leq 8 \IC_\mu(\pi) / \delta$.  Let us say that a
    transcript is ``bad for Alice'' (resp. Bob) if it lies in the set
    $\Bad_\Delta(\Pi_{x,y}, \Pi_x)$ (resp. in the set
    $\Bad_\Delta(\Pi_{x,y}, \Pi_y)$.  Using \autoref{lem:typical},
    this implies that
    \begin{align*}
      \gamma_A & = \Pr[\Pi_{x,y} \text{ bad for Alice}]
      \leq \frac{\frac{8}{\delta} \IC_\mu(\pi) + 1}{\Delta} \leq
      \delta/4         \label{eq:naobnzp} \\
      \gamma_B & = \Pr[\Pi_{x,y} \text{ bad for Bob}]
      \leq \delta/4. 
    \end{align*}
    It follows that $\gamma = \Pr[\Pi_{x,y} \text{ bad for Alice or
      Bob}] \leq \gamma_A + \gamma_B \leq \delta / 2$.

    By definition, for any $j$, experiment $j$ is accepted if and
    only if $j \in \calA \cap \calB$.  Therefore, $\forall j \in [T]$:
    \begin{align*}
      \Pr[j \in \calB \mid j \in \calA] & =
      \Pr[\text{experiment $j$ is accepted} \mid j \in \calA ] \\
      & = \frac{\Pr[\text{experiment $j$ is accepted} ]}{\Pr[j \in \calA ]} \\ 
      & \geq \frac{1-\gamma}{2^\Delta} 
\\
      & \geq \frac{1-\frac{\delta}{2}}{2^\Delta}.
    \end{align*}
    where we used \autoref{step:alicebob} and
    \autoref{step:nofail} of \autoref{lem:single}, and
    the fact that $\neg B_D$ implies that $\gamma \leq \delta/4$.

    Also, observe that by the definition of the protocol, the choice
    of $h$ and $r$ are completely independent of the experiments.
    Therefore we may add the condition that $h(j) = r$ without
    altering the probability. Since $j \in \calA$ implies $\calA\neq \varnothing$, we can add this condition too.
    We use this with $j=i$, so therefore we may write
\[
 \Pr_{\alean}[G] \geq (1-\delta/8) \frac{1-\frac{\delta}{2}}{2^{k+\Delta}} \geq  \frac{1-\frac{5\delta}{8}}{2^{k+\Delta}}.
\]
  Finally, observe that $H \setminus B_C = G \setminus B_C$, therefore
  we may conclude that:
  $$\Pr[H] \geq \Pr[H \setminus B_C] = \Pr[G \setminus B_C] \geq
  \Pr[G] - \Pr[B_C] \geq (1 - \tfrac{11\delta}{16}) 2^{-k-\Delta}.$$

  \item We will again use the event $G$ as defined in the previous
    point.  We will again use the fact that:
  \begin{align*}
\Pr_{\alean}[G ] & = \Pr[\calA\neq \varnothing] \Pr[h(i) = r\mid\calA\neq \varnothing ] \Pr[i \in \calB \mid i \in \calA  \wedge h(i)=r \wedge \calA\neq \varnothing]\\
&\leq \Pr[h(i) = r\mid\calA\neq \varnothing ] \Pr[i \in \calB \mid i \in \calA  \wedge h(i)=r \wedge \calA\neq \varnothing].
  \end{align*}
As before the first factor is exactly $2^{-k}$ for any $i$.
  We may also write:
    \begin{align*}
      \Pr[i \in \calB \mid i \in \calA ] & =
      \Pr[\text{experiment $i$ is accepted} \mid i \in \calA ] \\
      & = \frac{\Pr[\text{experiment $i$ is accepted}]}{\Pr[i \in \calA]} \\
      & \leq \frac{1}{2^\Delta},
    \end{align*}
    where we used \autoref{step:alicebob} and \autoref{step:nofail} of
    \autoref{lem:single}.  As with the previous point, adding the
    conditions $h(i) = r$ and $\calA\neq\varnothing$ does not affect
    the probabilities.  Therefore, $\Pr[G] \leq 2^{-k-\Delta}$.
    Finally, observe that $H \subseteq G \cup B_C$, and therefore:
    $$\Pr[H] \leq \Pr[G \cup B_C] \leq \Pr[G] + \Pr[B_C] \leq (1 +
    \tfrac{\delta}{16}) 2^{-k-\Delta}.$$
    
  \item The distribution of $\Pi'_{x,y}$ conditioned on not aborting
    \emph{and} on no collision is simply the distribution of the output
    of a single experiment, and we know from the facts about the
    single experiment (\autoref{lem:single}) that this is close to
    $\Pi_{x,y}$.  We wish to conclude that $\Pi'_{x,y}$ conditioned
    \emph{only} on not aborting is also close to $\Pi_{x,y}$.  The
    following lemma allows us to do this by using the fact that the
    probability of collision is small:
  \begin{claim}
    \label{lem:conditional}
    Let $\Pi$ and $\Pi'$ two distributions taking output in a common
    universe. Let $F$ and $H$ be two events in
    the underlying probability space of $\Pi'$. Finally, we let $\Pi'_E$ denote the distribution $\Pi'$ conditioned on $E=H\setminus F$, and assume that $|\Pi'_E - \Pi| \leq c$.
  Then it holds that $|\Pi'_H - \Pi| \leq c + \frac{\Pr_{\Pi'}[F]}{\Pr_{\Pi'}[H]}$.
  \end{claim}
  \begin{proof}
    For shorthand, for any event $E$ let us write $\Pi'_H(E) =
    \Pr_{\Pi'}[E \mid H]$, and similarly for $\Pi'(E)$ and $\Pi(E)$.
    For a set $S$ in the support of $\Pi'$ and $\Pi$, we let $S$ also
    denote the event that the value of the random variable is in $S$.

    It suffices to prove that, for all subsets $S$ in the union of the
    supports of $\Pi$ and $\Pi'_H$, it holds that $\Pi'_H(S) - \Pi(S)
    \leq c +\frac{\Pr_{\Pi'}[F]}{\Pr_{\Pi'}[H]}$.  To show this, we may write:
    \begin{align*}
      \Pi'_H (S) - \Pi(S) &  = \frac{\Pi'(H \cap S)}{\Pi'(H)} - \Pi(S)
      \\
      & \leq  \frac{\Pi'(E \cap S) + \Pi'((H\setminus E) \cap S)}{\Pi'(H)} - \Pi(S) \\
      & \leq  \frac{\Pi'(E)(c+ \Pi(S)) + \Pi'(F)}{\Pi'(H)} -\Pi(S)
    \end{align*}
   since $(H\setminus E) \subseteq F$ and $|\Pi'_E - \Pi| \leq c$.  Using the fact that $E
    \subset H$, we can conclude that:
    $$ \Pi'_H(S) - \Pi(S) \leq c+ \frac{\Pi'(F)}{\Pi'(H)} = c +
    \frac{\Pr_{\Pi'}[F]}{\Pr_{\Pi'}[H]}.$$  
 \end{proof}
   
   We apply this lemma with $\Pi = \Pi_{x,y}$, $\Pi' = \Pi'_{x,y}$,
   $F = B_C$, and $H$ the event that $\pi'(x,y) \neq
   \bot$. We note that $E=H\setminus B_C=G\setminus B_C$. We calculate $c$, $\Pr_{\alean}[F]$ and  $\Pr_{\alean}[H]$:
  \begin{itemize}
  
   \item Since $h$ and $r$ are completely independent of the actual
     experiments themselves, it holds that the distribution of
     $\Pi'_{x,y}$ conditioned on $G\setminus B_C$ is identical to the output of
     a single experiment (\autoref{alg:experiment}). Since $(x, y)
     \notin B_D$, the measure of the bad set $\Bad_\Delta(\Pi_{x,y},
     \Pi_y) \cup \Bad_\Delta(\Pi_{x,y}, \Pi_x)$ is bounded by
     $\delta/2$.  We apply \autoref{step:stat} of \autoref{lem:single}
     to deduce that $| \Pi_{x,y} - \Pi'_{x, y,E} | \leq \delta/2 = c$.
    
 \item From \autoref{item:collision} of \autoref{claim:compresspf}, we know that $\Pr_{\alean}[F]= \Pr_{\alean}[B_C] \leq \frac{\delta}{16} \cdot 2^{-(k+\Delta)}$.

 \item From \autoref{item:gooda} of \autoref{claim:compresspf}, we
   know $\Pr_{\alean}[H] \geq (1 - \frac{11\delta}{16})
   2^{-(k+\Delta)}$ because $(x, y) \notin B_D$.
 \end{itemize}
 
 Therefore, \autoref{lem:conditional} implies that:
 $$| \Pi_{x,y} - \Pi'_{x, y,\neq \bot} | \leq \delta/2 +
 \frac{\delta}{16(1- \frac{11\delta}{16})} \leq \delta/2 + \delta/5 <
 3\delta/4$$ 
 where we use the assumption that $\delta \leq 1$.
\end{enumerate}
\end{proof}

\subsection{Proof of the Compression Lemma}
 
\begin{proof}[Proof of \autoref{lem:compress}]
  Set $\lambda = 2^{-(k+\Delta)}$.  It holds that $\lambda \geq 2^{-C
    (\IC_\mu(\pi) / \delta^2 + 1/\delta)}$ for $C=64$.
Let $\rel$ be any subset of the support of $(X,Y,\pi(X,Y))$.
Then
  \begin{align*}
    \Pr_{\alea,(X,Y) \sim \mu}\![(X, Y, \pi(X,Y))\in \rel] \leq \!\Pr_{(X,Y)\sim\mu}\![B_D] + \!\Pr_{(X,Y)\sim\mu}\![\neg B_D]
    \cdot \!\Pr_{\alea, (X,Y)\sim\mu}\![(X, Y, \pi(X,Y))\in \rel \mid \neg
    B_D].
  \end{align*}

  Applying \autoref{item:statdist} of \autoref{claim:compresspf} and the
  fact that $\rel$ is simply an event, it follows that for all $(x,y) \notin B_D$ 
  \begin{equation*}
    \Pr_{\alea}[(x, y, \pi(x, y)) \in \rel] \leq \Pr_{\alean}[(x, y, \pi'(x, y)) \in \rel
    \mid \pi'(x, y) \neq \bot] + \tfrac{3\delta}{4}.
  \end{equation*}

  Since $\Pr[B_D] \leq \delta/4$ (\autoref{item:largediv} of \autoref{claim:compresspf}),

  \begin{equation*}
    \Pr_{\alea, (X,Y)\sim \mu}[(X, Y, \pi(X, Y)) \in \rel] \leq \Pr_{\alean, (X,Y)\sim \mu}[(X, Y, \pi'(X, Y)) \in \rel
    \mid \pi'(x, y) \neq \bot] + \delta.
  \end{equation*}

  This proves one direction of \autoref{eq:compressdist}.  
  For the other direction, we have that  
  \begin{align}
&    \Pr[(X, Y, \pi'(X,Y))\in \rel \mid \pi'(X, Y) \neq \bot]\\ 
&\qquad \leq
    \Pr[B_D]+ \Pr[\neg B_D] \cdot \Pr[(X, Y,\pi'(X,Y))) \in \rel \mid
    \neg B_D, \pi'(X,Y) \neq \bot ] \nonumber\\
&\qquad     \leq \tfrac{\delta}{4} + \Pr[ \neg B_D] \cdot  \left( \Pr[(X, Y,
      \pi(X,Y)) \in \rel \mid \neg B_D] + \tfrac{3\delta}{4}
    \right) \label{eq:ghpoi} \\
&\qquad     \leq \Pr[(X, Y, \pi(X, Y)) \in \rel \wedge \neg B_D] +
    \delta \nonumber\\
&\qquad      \leq \Pr[(X, Y, \pi(X, Y)) \in \rel] + \delta
\nonumber 
 \end{align} %
 where in \autoref{eq:ghpoi} we applied \autoref{item:statdist} of
 \autoref{claim:compresspf}. 
This proves \autoref{eq:compressdist} of \autoref{lem:compress}. 


\autoref{eq:compressbadbot} follows immediately
from \autoref{item:nottoogood} of \autoref{claim:compresspf}.
%

Finally, for \autoref{eq:compressavgbot}, we may write:
\begin{eqnarray*}
  \Pr_{\alean,(X,Y)\sim\mu}[\pi'(X, Y) \neq \bot] &\geq& \Pr_{(X,Y)\sim\mu}[\neg B_D]
  \Pr_{\alean,(X,Y)\sim\mu}[\pi'(X, Y) \neq \bot|\neg B_D]\\ 
  &\geq & (1-\tfrac{\delta}{4})(\Pr_{\alean,(X,Y)\sim\mu}[H \mid \neg B_D])\\
  &\geq & (1-\tfrac{\delta}{4})(1 - \tfrac{11\delta}{16})\lambda \quad
  > \quad (1-\delta)\lambda 
\end{eqnarray*}
where we used \autoref{item:gooda} of \autoref{claim:compresspf}.
\end{proof}

\section{Applications}

We can prove lower bounds on the information complexity of specific problems, by checking that their communication lower bounds were obtained by one of the
methods subsumed by the relaxed partition bound, including the
factorization norm, smooth rectangle, rectangle, or discrepancy.
However, a bit of care is required to ascertain this.  For example,
while a paper may say it uses the ``rectangle bound'', we must still
verify that the value of the linear program for $\bprt$ (or one of the
subsumed programs such as $\srec$ or $\rect$) is at least the claimed
bound, since different authors may use the term ``rectangle bound'' to
mean different things. In particular what they call ``rectangle
bound'' may not satisfy the constraints of the rectangle/smooth
rectangle linear programs given by Jain and Klauck \cite{JK10}.  After
we have verified that $\bprt$ is appropriately bounded, then we can
apply our main theorem (\autoref{thm:mainthm}).  We do this for the
problems below.

\subsection{Exponential separation of quantum communication and classical information complexity}

We prove that the quantum communication complexity of the Vector in 
Subspace Problem is exponentially smaller than its classical information 
complexity~(\autoref{VSP}).
In the Vector in Subspace Problem $\VSP_{0,n}$, Alice is given an $n/2$ dimensional subspace of
an $n$ dimensional space over $\R$, and Bob is given a vector. 
This is a partial function, and the promise is that either Bob's vector
lies in the subspace, in which case the function evaluates to $1$, or it lies in the orthogonal subspace, in which case the function evaluates to $0$. Note that the input set of $\VSP_{0,n}$ is continuous, but it can be discretized by rounding, which leads to the problem $\tVSP_{\theta,n}$ (see~\cite{KR11} for details).

Klartag and Regev~\cite{KR11} show that the Vector in Subspace Problem
can be solved with an $O(\log n)$ quantum protocol, but the 
randomized communication complexity of this problem is 
$\Omega(n^{1/3})$. Their lower bound uses a modified version of the rectangle bound, which can be shown to be still weaker than our relaxed partition bound.

\begin{lemma}
\label{lem:VSP}
There exist universal constants $C$ and $\gamma$ such that for any $\eps$,
\[\bprt_\eps(\tVSP_{\theta,n})\geq\frac{1}{C}(0.8-2.8\eps)\exp (\gamma n^{1/3}).\]
\end{lemma}

Klartag and Regev's lower bound is based on the following lemma:
\begin{lemma}[\cite{KR11}]\label{lem:klartag-regev}
Let $f=\VSP_{0,n}$, $\calX$ and $\calY$ denote Alice and Bob's input sets for $f$, $\sigma$ be the uniform distribution over $\calX\times\calY$ and $\sigma_b$ be the uniform distribution over $f^{-1}(b)$. There exist universal constants $C$ and $\gamma$ such that, for any rectangle $R$ and any $b\in\zo$, we have:
\begin{align*}
 \sigma_b(R\cap f^{-1}(b))\geq 0.8\sigma(R)-C\exp(-\gamma n^{1/3}).
\end{align*}
\end{lemma}

We show that this implies the required lower bound on the relaxed partition bound of $\tVSP_{\theta,n}$.

\begin{proof}[Proof of \autoref{lem:VSP}]
Let us first consider $f=\VSP_{0,n}$ and show that $\log\bprt_\eps(\VSP_{0,n})=\Omega(n^{1/3})$. Note that since the input set of $\VSP_{0,n}$ is continuous, we need to extend the definition of the relaxed partition bound to the case of a continuous input set, but this follows naturally by replacing summation symbols by integrals.

Let $\mu=\frac{\sigma_0+\sigma_1}{2}$, which satisfies $\mu(f^{-1})=1$, \ie $\mu$ only has support on valid inputs. Then, Lemma~\ref{lem:klartag-regev} implies that for any rectangle $R$ and any $b\in\zo$, we have
\begin{align*}
 2\mu(R\cap f^{-1}(b \oplus 1))\geq 0.8\sigma(R)-C\exp(-\gamma n^{1/3}).
\end{align*}
Since $\mu$ only has support on $f^{-1}$ and $f$ is Boolean, we have $\mu(R)=\mu(R\cap f^{-1}(b))+\mu(R\cap f^{-1}(b\oplus 1))$, so that the inequality can be rewritten as
\begin{align*}
 2\mu(R)-2\mu(R\cap f^{-1}(b))\geq 0.8\sigma(R)-C\exp(-\gamma n^{1/3}).
\end{align*}
Setting $t=\frac{\exp(\gamma n^{1/3})}{C}$, we have that:
\begin{align*}
 t\cdot \left( 2\mu(R\cap f^{-1}(b))+0.8\sigma(R)-2\mu(R)\right) \leq 1.
\end{align*}

We now construct a feasible point for the dual formulation of $\bprt_\eps(f)$ given in \autoref{claim:bprt-lp} by setting $\alpha(x,y)=2t\mu(x,y)$ for $(x,y)\in f^{-1}$ and $\alpha(x,y)=0.8t\sigma(x,y)$ otherwise; and $\beta(x,y)=2t\mu(x,y)-0.8t\sigma(x,y)$ for $(x,y)\in f^{-1}$, and $\beta(x,y)=0$ otherwise. Note that for $(x,y)\in f^{-1}$, we have $\mu(x,y)=\sigma(x,y)/\sigma(f^{-1})\geq\sigma(x,y)$, hence all these values are positive. By the previous inequality, we also have
\begin{align*}
 \int_{R \cap f^{-1}(b)} \alpha(x,y) dxdy+ \int_{R \setminus f^{-1}}
\alpha(x,y) dxdy- \int_{R}\beta(x,y) dxdy \leq 1
\end{align*}
for any $R,b$, therefore this is a valid feasible point. Moreover, the corresponding objective value is
\begin{align*}
&  \int_{\calX\times\calY} (1 - \eps) \alpha(x,y) dxdy-
  \int_{\calX\times\calY} \beta(x,y) dxdy\\
&\qquad\qquad =t\left( (1-\eps)2\mu(f^{-1})+(1-\eps)0.8(1-\sigma(f^{-1}))-2\mu(f^{-1})+0.8\sigma(f^{-1})\right)\\
&\qquad\qquad=t(0.8-2.8\eps+0.8\eps\sigma(f^{-1}))\\
&\qquad\qquad\geq t(0.8-2.8\eps).
\end{align*}
Finally, let us note that we can construct a zero-communication protocol for $\VSP_{0,n}$ by first rounding off the inputs, and then applying a zero-communication protocol for $\tVSP_{\theta,n}$. This means that we can turn a feasible point for 
the primal form of $\bprt_\eps(\tVSP_{\theta,n})$ (given in \autoref{claim:bprt-lp}) into a feasible point for $\bprt_\eps(\VSP_{0,n})$ with the same objective value, so that $\bprt_\eps(\tVSP_{\theta,n})\geq\bprt_\eps(\VSP_{0,n})\geq t(0.8-2.8\eps)$. 
\end{proof}

Finally, \autoref{lem:VSP} together with \autoref{thm:mainthm} implies that $\IC(\tVSP_{\theta,n},\eps)=\Omega(n^{1/3})$ and also \autoref{VSP}.

This allows us to conclude that
the information complexity of this function is at least $\Omega(n^{1/3})$.
This solves Braverman's Open Problem~3 
(Are there problems for which $Q(f, \eps) = O(\mathrm{polylog}(\IC(f, \eps)))$?)
and Open Problem~7 (is it true that $\IC(\tVSP_{\theta,n}, 1/3) = n^{\Omega(1)}$?)

Moreover, our result implies an exponential separation between classical and quantum information complexity. We refrain from defining quantum information cost and complexity in this paper (see \cite{JN10} for a definition), but since the quantum information complexity is always smaller than the quantum communication complexity, the separation follows trivially from \autoref{VSP}. 

\subsection{Information complexity of the Gap Hamming Distance Problem}
\label{sec:ghd}
We prove that the information complexity of Gap Hamming Distance
is $\Omega(n)$ (\autoref{thm:GHD}; Open Problem~6 in~\cite{Bra11b}).
In the Gap Hamming Distance Problem ($\GHD_n$), Alice and Bob each receive 
a string  of length $n$ and they need to determine whether their Hamming distance is at least
$n/2 +\sqrt{n}$ or less than $n/2-\sqrt{n}$. 
We prove that the information complexity of Gap Hamming Distance
is $\Omega(n)$ (\autoref{thm:GHD}; Open Problem~6 in~\cite{Bra11b}).
The communication complexity of Gap Hamming Distance was shown to be
$\Omega(n)$ by Chakrabarti and Regev~\cite{CR10}.  The proof was subsequently
simplified by Vidick~\cite{Vid11} and Sherstov~\cite{She11}.  The first
two proofs use the smooth rectangle bound, while Sherstov uses the rectangle/corruption bound.

The corruption bound used by Sherstov is a slight refinement of the
rectangle bound as defined by Jain and Klauck~\cite{JK10},
since it can handle distributions that put small weight on
the set of inputs that map to some function value $z$.  
It can be shown that this bound is weaker than our relaxed partition bound, which implies  \autoref{thm:GHD}.

\begin{lemma}
\label{lem:GHD}
There exist universal constants $C$ and $\delta$ such that for any small
enough $\eps$,
$\bprt_\eps(\GHD)\geq C2^{\delta n}$.
\end{lemma}

We let the output set be $\setout = \{-1, 1\}$ rather
than bits, to follow the notation of \cite{She11}.  Let us recall the
corruption bound used by Sherstov.
\begin{theorem}
For any function $f$ with output set $\setout = \{-1, 1\}$, 
and $\epsilon, \delta, \beta>0$
if a distribution on the inputs $\mu$ is such that
$$\mu(R) > \beta \implies \mu(R\cap f^{-1}(1))>\delta \mu(R\cap f^{-1}(-1))$$then
$$2^{R_\epsilon(f)} \geq \frac{1}{\beta}(\mu(f^{-1}(-1)) - \frac{\epsilon}{\delta}).$$
\end{theorem}

We can derive the general corruption bound as used by Sherstov by
giving a feasible solution to the dual of the linear program by Jain
and Klauck.  In the dual form, $\rect_\eps^{z}(f)$ is defined for $z
\in \setout$ as
\begin{align}
\rect_\eps^{z}(f)=\max_{\alpha_{x,y}\geq 0} & 
	(1 - \eps)\sum_{(x,y)\in f^{-1}(z)} \alpha_{x,y} 
	- \eps\sum_{(x,y)\in f^{-1}\setminus f^{-1}(z)} \alpha_{x,y} \\
\forall R, & \quad 
	\sum_{(x,y)\in R\cap f^{-1}(z)} \alpha_{x,y} 
	- \sum_{(x,y)\in R\cap(f^{-1}\setminus f^{-1}(z))} \alpha_{x,y} \leq 1.
\end{align}

For the remainder of this section fix $z = -1$.  Consider the
following assignment for $\rect_\eps^{z}$, where 
$\mu$ is the distribution over the inputs in the theorem.
Letting 
 $ \alpha_{x,y} =  \frac{1}{\beta}\mu(x,y)$ if $f(x,y)=-1$,
and 
$ \alpha_{x,y} =  \frac{1}{\delta}\frac{1}{\beta}\mu(x,y)$ if $f(x,y)=1$,
we can verify the constraints and the objective value is
greater than the corruption bound.

To conclude the bound on $\GHD$, Sherstov gives a reduction to the Gap Orthogonality Problem ($\ORT$)
and proves the following lemma.
\begin{lemma}[~\cite{She11}]
Let $f$ denote the Gap Orthogonality Problem.
For a small enough constant $\delta<1$ and 
for the uniform distribution $\mu$ on inputs, and any rectangle
$R$ such that $\mu(R) > 2^{-\delta n}$,
$\mu(R\cap f^{-1}(1))>\delta \mu(R)$.
Furthermore, $\mu(f^{-1}(-1))=\Theta(1)$.
\end{lemma}

Putting all this together, we have that
$\rect_\eps^{z}(\ORT) \geq C  2^{\delta n}$ for appropriate choices of
$\eps$.

Finally, there is a simple reduction from one instance of $\ORT$ to
two instances of $\GHD$ (see~\cite{She11} for details).  This means
that $[\bprt_{\eps/2}(\GHD)]^2 \geq \bprt_{\eps}(\ORT)$, since given
any zero-communication protocol for $\GHD$ with error $\eps / 2$ and
efficiency $\eta$, one can give a zero-communication protocol for
$\ORT$ in the obvious way, by simply running the protocol for $\GHD$
twice and the reduction to determine the output for $\ORT$.  By a
union bound this incurs error $\eps$, and it has efficiency $\eta^2$ (the
probability that the two independent calls to the $\GHD$ protocol both
do not abort).  Therefore we may conclude that
$[\bprt_{\eps/2}(\GHD)]^2 \geq \bprt_{\eps}(\ORT) \geq
\rect_\eps^{z}(\ORT) \geq C 2^{\delta n}$.

\section{Conclusions and open problems}

We have shown that the information complexity is lower bounded by a relaxed version of the partition bound.  This subsumes all known
algebraic and rectangle based methods, except the partition bound. It
remains to be seen if the partition bound also provides a lower bound on the information complexity. 
Alternatively, if we would like to separate the communication and information complexities, then possible candidates could be functions whose partition bound is strictly larger than their relaxed partition bound.

Moreover, we have seen how the relaxed partition bound naturally relates to zero-communication protocols with abort. Actually, we can relate all other lower bound methods to different variants of zero-communication protocols~\cite{DKLR11, LLR12}. This provides new insight on the inherent differences between these bounds and may lead to new lower bound methods, coming from different versions of zero-communication protocols. Moreover, since these protocols have been extensively studied in the field of quantum information, it is intriguing to see what other powerful tools can be transferred to the model of classical communication complexity.


\section{Acknowledgements}

We would like to thank Amit Chakrabarti and Oded Regev for helpful
comments. J.R. acknowledges support from the action \emph{Mandats de
  Retour 2010} of the \emph{Politique Scientifique F\'ed\'erale
  Belge.} This research was funded in part by the EU grant QCS, ANR
Jeune Chercheur CRYQ, ANR Blanc QRAC (ANR-08-EMER-012), and EU ANR
Chist-ERA DIQIP.
\bibliographystyle{alphaurl}
\bibliography{ic-rec}

\appendix


\section{Relaxed partition bound and other bounds}\label{app:bprt}
Recall the definition of the partition bound $\prt_\eps(f)$ in~\cite{JK10}:
\begin{definition}[\cite{JK10}]
The partition bound $\prt_\eps(f)$ is defined as the value of the following linear program:
\begin{align*}
\prt_\eps(f)=\min_{w_{R,z}\geq 0}  \sum_{R,z} w_{R,z} \quad \text{ subject to:}\quad
 &\forall (x,y)\in \setin,\quad \sum_{R:(x,y)\in R} w_{R,f(x,y)} \geq 1-\eps\quad \nonumber\\
 &\forall (x,y)\in\calX\times\calY,\quad \sum_{z,R:(x,y)\in R} w_{R,z} =1.
\end{align*}
\end{definition}

(To put this into a form similar to \autoref{def:relaxed-prt}, simply
substitute $\frac{1}{\eta} = \sum_{R,z} w_{R,z}$ and $p_{R,z} =
w_{R,z} \eta$.)

We first show that $\bprt_\eps(f)$ is indeed a relaxation of the partition bound by providing linear programming formulations for $\bprt_\eps^\mu(f)$ and $\bprt_\eps(f)$.
\begin{claim}\label{claim:bprt-lp}
$\bprt_\eps^\mu(f)$ can be expressed as the following linear programs:
\begin{enumerate}
 \item Primal form:
\begin{align*}
\bprt_\eps^\mu(f)=\min_{w_{R,z}\geq 0} & \sum_{R,z} w_{R,z} \quad \text{ subject to:}
\nonumber\\
 & \sum_{(x,y)\in \setin}\mu_{x,y}\sum_{R:(x,y)\in R} w_{R,f(x,y)}+\sum_{(x,y)\notin \setin}\mu_{x,y}\sum_{z,R:(x,y)\in R} w_{R,z} \geq 1-\eps \nonumber\\
 &\forall (x,y)\in\calX\times\calY,\quad \sum_{z,R:(x,y)\in R} w_{R,z} \leq 1,
\end{align*}
 \item Dual form:
\begin{align*}
\bprt_\eps^\mu(f)=\max_{\alpha\geq 0, \beta_{x,y}\geq 0} & (1 - \eps)\alpha -
  \sum_{(x, y)\in\calX\times\calY} \beta_{x,y} \text{ subject
    to:} \nonumber\\
\forall R,z, & \quad \sum_{(x, y) \in R \cap f^{-1}(z)}
\alpha\mu_{x,y} + \sum_{(x, y) \in R \setminus \setin}
\alpha\mu_{x,y} - \sum_{(x, y) \in R}\beta_{x,y} \leq 1.
\end{align*}
\end{enumerate}
Similarly, $\bprt_\eps(f)$ is given by the value of the following linear programs:
\begin{enumerate}
 \item Primal form:
\begin{align*}
\bprt_\eps(f)=\min_{w_{R,z}\geq 0}  \sum_{R,z} w_{R,z} \quad \text{ subject to:}\quad
 &\forall (x,y)\in \setin,\quad \sum_{R:(x,y)\in R} w_{R,f(x,y)} \geq 1-\eps \nonumber\\
 &\forall (x,y)\notin \setin,\quad \sum_{z,R:(x,y)\in R} w_{R,z} \geq 1-\eps \nonumber\\
 &\forall (x,y)\in\calX\times\calY,\quad \sum_{z,R:(x,y)\in R} w_{R,z} \leq 1,
\end{align*}
 \item Dual form:
\begin{align*}
\bprt_\eps(f)=\max_{\alpha_{x,y}\geq 0, \beta_{x,y}\geq 0} & \sum_{(x, y)\in\calX\times\calY} (1 - \eps) \alpha_{x,y} -
  \sum_{(x, y)\in\calX\times\calY} \beta_{x,y} \text{ subject
    to:} \nonumber\\
\forall R,z, & \quad \sum_{(x, y) \in R \cap f^{-1}(z)} \alpha_{x,y} + \sum_{(x, y) \in R \setminus \setin}
\alpha_{x,y} - \sum_{(x, y) \in R}\beta_{x,y} \leq 1.
\end{align*}
\end{enumerate}
\end{claim}
\begin{proof}
The primal form of $\bprt_\eps^\mu(f)$ can be obtained from \autoref{def:relaxed-prt} by using the change of variables $w_{R,z}=p_{R,z}/\eta$. The dual form then follows from standard linear programming duality.

As for $\bprt_\eps(f)$, by \autoref{def:relaxed-prt} we have $\bprt_\eps(f)=\max_\mu\bprt_\eps^\mu(f)$. The dual form of $\bprt_\eps(f)$ then follows from the dual form of $\bprt_\eps^\mu$ via the change of variables $\alpha_{x,y}=\alpha\mu_{x,y}$. Finally, the primal form can be obtained via standard linear programming duality.
\end{proof}
Let us compare the primal form of $\bprt_\eps(f)$ with the definition of the partition bound $\prt_\eps(f)$. We see that the first two constraints in $\bprt_\eps(f)$ imply that $\sum_{z,R:(x,y)\in R} w_{R,z}$ lies between $1-\eps$ and $1$, while for $\prt_\eps(f)$, this should be exactly equal to $1$.  In terms of zero-communication protocols, this difference can be interpreted as follows: for $\prt_\eps(f)$, the protocol should output anything but $\bot$ with constant probability $\eta$ for any input $(x,y)$, while for $\bprt_\eps(f)$, the probability of not outputting $\bot$ is allowed to fluctuate between $(1-\eps)\eta$ and $\eta$.

Just as the partition bound, the relaxed partition bound is stronger than the smooth rectangle bound, defined in~\cite{JK10} as follows:
\begin{definition}[\cite{JK10}] The smooth rectangle bound $\srec_\eps^{z_0}(f)$ is the value of the following linear program:
\begin{align*}
\srec_\eps^{z_0}(f)=\min_{w'_{R}\geq 0}  \sum_{R} w'_{R} \quad \text{ subject to:}\quad
 &\forall (x,y)\in f^{-1}(z_0),\quad \sum_{R:(x,y)\in R} w'_{R} \geq 1-\eps \nonumber\\
 &\forall (x,y)\in f^{-1}(z_0),\quad \sum_{R:(x,y)\in R} w'_{R} \leq 1 \nonumber\\
 &\forall (x,y)\in \setin\setminus f^{-1}(z_0),\quad \sum_{R:(x,y)\in R} w'_{R} \leq \eps.
\end{align*}
\end{definition}

Let us now prove \autoref{claim:srec-bprt-prt}, that is,  $\srec^{z_0}_\eps(f)\leq \bprt_\eps(f)\leq\prt_\eps(f)$.

\begin{proof}[Proof of \autoref{claim:srec-bprt-prt}]
 The second inequality is immediate since the only difference between the linear programs is that the constraint $\sum_{z,R:(x,y)\in R} w_{R,z} =1$ in the primal form of $\prt_\eps(f)$ has been relaxed to $1-\eps\leq\sum_{z,R:(x,y)\in R} w_{R,z} \leq 1$.

As for the first inequality, let $w_{R,z}$ be an optimal solution for the primal formulation of $\bprt_\eps(f)$ in \autoref{claim:bprt-lp}. Then, it is straightforward to check that setting $w'_R=w_{R,z_0}$ leads to a feasible point for $\srec_\eps^{z_0}(f)$, with objective value $\sum_{R} w'_{R}=\sum_R w_{R,z_0}\leq \sum_{R,z} w_{R,z}=\bprt_\eps(f)$.
\end{proof}

\end{document}